\documentclass[11pt,letterpaper]{article}
\usepackage[english]{babel}
\usepackage{dsfont}
\usepackage[utf8x]{inputenc}
\usepackage[T1]{fontenc}
\usepackage{authblk}

\usepackage[margin=1in]{geometry}

\usepackage[square,numbers]{natbib}
\usepackage{amsmath}
\usepackage{graphicx}
\usepackage{hyperref}
\usepackage{bbm,microtype,mathrsfs,amsmath,amssymb,color,amsthm,graphicx,cleveref,Quantum,times,bm}
\usepackage{pgf,tikz}
\usepackage{tikz-cd}
\usepackage{times}
\usetikzlibrary{through,calc,positioning,patterns}
\newtheorem{thm}{Theorem}

\newtheorem{prop}{Proposition}[thm]

\usepackage{xcolor}
\input{Qcircuit}	

\newcommand{\CA}{\mathcal{A}}
\newcommand{\CB}{\mathcal{B}}

\newcommand{\CD}{\mathcal{D}}

\newcommand{\CH}{\mathcal{H}}

\newcommand{\CJ}{\mathcal{J}}

\newcommand{\CM}{\mathcal{M}}
\newcommand{\CN}{\mathcal{N}}

\newcommand{\CP}{\mathcal{P}}

\newcommand{\CZ}{\mathcal{Z}}

\newcommand{\drawRT}[6]{
\pgfmathparse{#2-#1==180};  
\ifnum\pgfmathresult>0 \draw [very thick, shorten >= -0.5, shorten <= -0.5] (#2:#3)--(#1:#3);  
	\draw [very thick, shorten >= -0.6, shorten <= -0.6] (#1:#3) arc (#1:#2:#3); 

\else
	\coordinate (O) at (0,0);
	\coordinate (A) at (#1:#3);
	\coordinate (B) at (#2:#3);
	\coordinate (Ap) at ([shift={({#1-90}:#3)}]A);
	\coordinate (App) at ([shift={({#1+90}:#3)}]A);
	\coordinate (Bp) at ([shift={({#2-90}:#3)}]B);
	\coordinate (Bpp) at ([shift={({#2+90}:#3)}]B);
	\coordinate (X) at (intersection of Ap--App and Bp--Bpp);
	\pgfmathparse{#4==1};
	\ifnum\pgfmathresult>0 \draw [very thick, shorten >= -0.6, shorten <= -0.6] (A) arc (#1:#2:#3); \fi
	\path (X); \pgfgetlastxy{\XCoord}{\YCoord}; 
	\pgfmathsetmacro{\rotateamt}{\ifdim\XCoord<0pt {atan(\YCoord/\XCoord)} \else \ifdim\XCoord>0pt {atan(\YCoord/\XCoord)+180} \else {-\YCoord/abs(\YCoord)*90}\fi \fi}

	\begin{scope}[rotate=\rotateamt]
		\pgfmathparse{#5==1}; 
		\ifnum\pgfmathresult>0 \fill[gray, opacity=0.2, #6]
		let
	 		\p1 = (A),
			\p2 = (B),
	 		\p3 = (X),
			\p4 = (A |- X), 
			\p5 = (B |-X), 
			\n1 = {veclen((\x3-\x1),(\y3-\y1))}, 
			\n2 = {\y1-\y3}, 
			\n3 = {\y2-\y3}  
			in
			([shift=({#1-\rotateamt}:#3)]O) arc ({#1-\rotateamt}:{#2-\rotateamt}:#3) -- ([shift=({asin(\n3/\n1)}:\n1)]X) arc ({asin(\n3/\n1)}:{asin(\n2/\n1)}:\n1);\fi 

		\draw [very thick, shorten >= 0.0, shorten <= 0.0]
			let
	 		\p1 = (A),
			\p2 = (B),
	 		\p3 = (X),
			\p4 = (A |- X), 
			\p5 = (B |-X), 
			\n1 = {veclen((\x3-\x1),(\y3-\y1))}, 
			\n2 = {\y1-\y3}, 
			\n3 = {\y2-\y3}  
			in
			([shift=({asin(\n3/\n1)}:\n1)]X) arc ({asin(\n3/\n1)}:{asin(\n2/\n1)}:\n1); 
	\end{scope}
\fi
}

\newcommand{\drawRTdashed}[6]{
\pgfmathparse{#2-#1==180};  
\ifnum\pgfmathresult>0 \draw [dashed, shorten >= -0.5, shorten <= -0.5] (#2:#3)--(#1:#3);  
	\draw [dashed, shorten >= -0.6, shorten <= -0.6] (#1:#3) arc (#1:#2:#3); 

\else
	\coordinate (O) at (0,0);
	\coordinate (A) at (#1:#3);
	\coordinate (B) at (#2:#3);
	\coordinate (Ap) at ([shift={({#1-90}:#3)}]A);
	\coordinate (App) at ([shift={({#1+90}:#3)}]A);
	\coordinate (Bp) at ([shift={({#2-90}:#3)}]B);
	\coordinate (Bpp) at ([shift={({#2+90}:#3)}]B);
	\coordinate (X) at (intersection of Ap--App and Bp--Bpp);
	\pgfmathparse{#4==1};
	\ifnum\pgfmathresult>0 \draw [dashed, shorten >= -0.6, shorten <= -0.6] (A) arc (#1:#2:#3); \fi
	\path (X); \pgfgetlastxy{\XCoord}{\YCoord}; 
	\pgfmathsetmacro{\rotateamt}{\ifdim\XCoord<0pt {atan(\YCoord/\XCoord)} \else \ifdim\XCoord>0pt {atan(\YCoord/\XCoord)+180} \else {-\YCoord/abs(\YCoord)*90}\fi \fi}

	\begin{scope}[rotate=\rotateamt]
		\pgfmathparse{#5==1}; 
		\ifnum\pgfmathresult>0 \fill[gray, opacity=0.2, #6]
		let
	 		\p1 = (A),
			\p2 = (B),
	 		\p3 = (X),
			\p4 = (A |- X), 
			\p5 = (B |-X), 
			\n1 = {veclen((\x3-\x1),(\y3-\y1))}, 
			\n2 = {\y1-\y3}, 
			\n3 = {\y2-\y3}  
			in
			([shift=({#1-\rotateamt}:#3)]O) arc ({#1-\rotateamt}:{#2-\rotateamt}:#3) -- ([shift=({asin(\n3/\n1)}:\n1)]X) arc ({asin(\n3/\n1)}:{asin(\n2/\n1)}:\n1);\fi 

		\draw [dashed, shorten >= 0.0, shorten <= 0.0]
			let
	 		\p1 = (A),
			\p2 = (B),
	 		\p3 = (X),
			\p4 = (A |- X), 
			\p5 = (B |-X), 
			\n1 = {veclen((\x3-\x1),(\y3-\y1))}, 
			\n2 = {\y1-\y3}, 
			\n3 = {\y2-\y3}  
			in
			([shift=({asin(\n3/\n1)}:\n1)]X) arc ({asin(\n3/\n1)}:{asin(\n2/\n1)}:\n1); 
	\end{scope}
\fi
}

\title{Entanglement Wedge Reconstruction using the Petz Map}
\author[1]{Chi-Fang Chen \thanks{\tt chifangc@stanford.edu}}
\author[2]{Geoffrey Penington \thanks{\tt geoffp@stanford.edu}}
\author[1,2,3]{Grant Salton \thanks{\tt gsalton@caltech.edu}}
\affil[1]{\small \em Department of Physics, Stanford University, Stanford CA 94305, USA}
\affil[2]{\small \em Stanford Institute for Theoretical Physics, Stanford University, Stanford CA 94305, USA}
\affil[3]{\small \em Institute for Quantum Information and Matter, Caltech, Pasadena, CA 91125, USA}

\begin{document}
\date{April 2, 2019}
\maketitle

\begin{abstract}
At the heart of recent progress in AdS/CFT is the question of subregion duality, or entanglement wedge reconstruction: which part(s) of the boundary CFT are dual to a given subregion of the bulk?  This question can be answered by appealing to the quantum error correcting properties of holography, and it was recently shown that robust bulk (entanglement wedge) reconstruction can be achieved using a universal recovery channel known as the \emph{twirled Petz map}.  In short, one can use the twirled Petz map to recover bulk data from a subset of the boundary. However, this map involves an averaging procedure over bulk and boundary modular time, and hence it can be somewhat intractable to evaluate in practice. We show that a much simpler channel, the Petz map, is sufficient for entanglement wedge reconstruction for any code space of fixed finite dimension -- no twirling is required. Moreover, the error in the reconstruction will always be non-perturbatively small. From a quantum information perspective, we prove a general theorem extending the use of the Petz map as a general-purpose recovery channel to subsystem and operator algebra quantum error correction.
\end{abstract}

\tableofcontents
\newpage

\section{Introduction}
The AdS/CFT correspondence -- a duality between a gravitational theory in asymptotically AdS space, and a conformal field theory in one fewer spatial dimension --  has enjoyed extraordinary success in the twenty years since it was first proposed~\cite{Maldacena1999}.  An important problem in AdS/CFT is that of subregion duality: which subregion of the CFT (if any) is dual to a given subregion of the bulk spacetime?  Recently, it was discovered that the bulk-to-boundary map in AdS/CFT defines a quantum error correcting code~\cite{almheiri2015bulk,pastawski2015holographic}.  In light of this profound revelation, the problem of subregion duality can be rephrased: which subregion of the bulk spacetime can be `reconstructed' from a given subregion of the boundary? 

Over the course of the last five years, this question has been answered~\cite{dong2016reconstruction,cotler2017entanglement,faulkner2017bulk}: the bulk region encoded in an arbitrary boundary region $A$ is the so-called \emph{entanglement wedge} of $A$, denoted $a$. Within a single, static timeslice, the entangement wedge\footnote{This definition is valid only within a single, static timeslice of a static bulk spacetime, or at a moment of time reversal symmetry. More generally, and more formally, the covariant Ryu-Takayanagi surface $\chi_A$ is defined to be the smallest surface of \emph{extremal} area homologous to $A$~\cite{hubeny2007covariant}. The entanglement wedge is then the bulk domain of dependence of \emph{any} achronal bulk surface bounded by $A$ and $\chi_A$. At higher orders in perturbation theory, one should use the quantum extremal surface, which extremizes the Ryu-Takayanagi formula $\CA/4 G_N + S_\text{bulk}$, where $S_\text{bulk}$ is the bulk entanglement entropy, rather than simply the classical area $\CA$~\cite{engelhardt2015quantum,dong2018entropy}.} of $A$ is the bulk region bounded by $A$ itself and its Ryu-Takayanagi surface $\chi_A$~\cite{ryu2006holographic}, which is the minimal area bulk surface anchored to the boundary of $A$; see Figure~\ref{EWre}. Given any bulk operator $\phi_a$ lying within the entanglement wedge $a$, there exists a boundary operator $\Phi_A$ acting only on the boundary region $A$, which approximately reproduces the action of the bulk operator $\phi_a$. The task of finding such an operator $\Phi_A$ is known as entanglement wedge reconstruction. 

The conjecture of entanglement wedge reconstruction was developed in~\cite{czech2012gravity, headrick2014causality, wall2014maximin} and established with increasing levels of rigour in~\cite{jafferis2016relative,dong2016reconstruction,faulkner2017bulk,cotler2017entanglement}. It was shown in~\cite{dong2018entropy,hayden2018learning} that the error in the reconstruction can be made non-perturbatively small at large boundary gauge group rank $N$, or equivalently small gravitational coupling $G_N$. 

The realization that bulk reconstruction can be rephrased in the language of quantum error correction~\cite{almheiri2015bulk} paved the way for most of this success. Bulk operators in AdS/CFT are only well defined within the ``code subspace'' $\mathcal{H}_{\text{code}}$ of states with the correct smooth bulk geometry. If we let $J : \mathcal{H}_{\text{code}} \to \mathcal{H}_{CFT}$ be the isometry embedding this code subspace into the larger CFT Hilbert space, entanglement wedge reconstruction can be rephrased as the task of finding a decoding channel $\mathcal{D}$ that can recover from the noisy channel $\mathcal{N} = [J(\cdot)J^\dagger]_{A}$, where $\rho_A$ is the restriction of the boundary state $\rho$ to region $A$. More specifically, entanglement wedge reconstruction is equivalent to the existence of a decoding channel $\mathcal{D}$ such that, for all states $\rho$ in the bulk code space,
\begin{align}
\mathcal{D} \circ \mathcal{N} (\rho) \approx \rho_a,\label{decode}
\end{align}
where $\rho_a$ is the restriction of the bulk state $\rho$ to the entanglement wedge $a$\footnote{Here, restriction can be thought of as a partial trace, but in an operator algebra quantum error correction picture, this is really a restriction of a state to a subalgebra.}. If such a decoding channel exists, then we can use the adjoint channel $\mathcal{D}^\dagger$, defined by 
\begin{align}
\tr [\mathcal{D}^\dagger(\phi) \,\sigma] = \tr[ \phi \,\mathcal{D}(\sigma)],
\end{align}
for all operators $\phi$ and states $\sigma$, to map bulk operators $\phi_a$ to boundary reconstructions $\Phi_A = \mathcal{D}^\dagger(\phi_a)$ with support only in region $A$. Since
\begin{align}
\tr(\Phi_A \,\rho) = \tr[\phi_a \,\mathcal{D} \circ \mathcal{N}(\rho)] \approx \tr(\phi_a \,\rho), \label{adjoint}
\end{align}
the expectation values of $\phi_a$ and $\Phi_A$ approximately agree for all states $\rho \in S(\mathcal{H}_{\text{code}})$. It can be shown that this is also true for higher point correlators~\cite{cotler2017entanglement}.

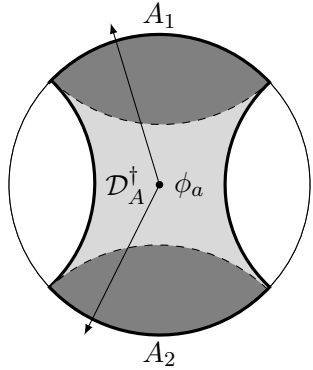
\begin{figure}[t]
\centering
\begin{tikzpicture}
	\fill [gray, opacity=0.3] (0,0) circle (2); 
	\draw [use as bounding box] (0,0) circle (2); 
	\drawRTdashed{43}{136}{2}{0}{1}{fill opacity=1, color=gray};
	\drawRT{136}{223}{2}{0}{1}{fill opacity=1, color=white};
	\drawRTdashed{223}{316}{2}{0}{1}{fill opacity=1, color=gray};
	\drawRT{-43}{43}{2}{0}{1}{fill opacity=1, color=white};
	\draw [very thick, shorten >= -0.6, shorten <= -0.6] (43:2) arc (43:137:2);
	\draw [very thick, shorten >= -0.6, shorten <= -0.6] (223:2) arc (223:317:2);
	
	\draw[-latex] (0,0) -- (-0.65,2.15);
	\draw (-0.45, 0) node {$\mathcal{D}^\dagger_{A}$};
	\draw[-latex] (0,0) -- (-1,-2);
	\draw (0.4,0) node{$\phi_{a}$};
	\fill (0,0) circle (0.05);
	\node at (0,2.25) {$A_1$};
	\node at (0,-2.25) {$A_2$};
	\draw [use as bounding box] (0,0) circle (2); 
\end{tikzpicture}
\vspace{0.3cm}
\caption{An operator $\phi_{a}$, acting on the entanglement wedge $a$ of $A = A_1 \cup A_2$, can be reconstructed on the boundary region $A$ by the map $\mathcal{D}^\dagger_{A}: M_{a} \rightarrow M_{A}$. The solid interior curves represent the RT surface of $A$ and the entire shaded region forms the entanglement wedge $a$ (restricted to a single timeslice).  The darker gray areas are the entanglement wedges of $A_1$ and $A_2$ individually, and also together form the causal wedge of $A$. Since the operator $\phi_a$ is not in the causal wedge of $A$, we cannot reconstruct it simply by using the bulk and boundary equations of motion; the more sophisticated machinery of quantum error correction is required. Moreover, $\phi_a$ can only be reconstructed on $A = A_1 \cup A_2$; neither $A_1$ nor $A_2$ alone contains any information about $\phi_a$.\label{EWre}}
\end{figure}

Interestingly, the entanglement wedge $a$ may contain regions outside of the `causal wedge' of $A$ (the intersection of the past and future of the boundary domain of dependence of $A$). Given a bulk operator $\phi$ in the causal wedge of a region $A$, it is well-understood how to reconstruct the operator $\phi$ within the boundary region $A$, given only the bulk and boundary equations of motion, using the so-called HKLL procedure~\cite{hamilton2008local,Morrison2014}. However, it was only by introducing the tool of quantum error correction that we have begun to understand that the entire entanglement wedge (rather than just the causal wedge) can indeed be reconstructed from region $A$.

The first clue that a boundary region $A$ encodes more than just the causal wedge actually comes from the Ryu-Takayanagi (RT) formula~\cite{ryu2006holographic,lewkowycz2013generalized}. Including the leading quantum correction~\cite{faulkner2013quantum}, the RT formula states that the entanglement entropy $S_A$ of any boundary region $A$ is given by 
\begin{align} \label{eq:RT}
S_A = \frac{\CA(\chi_A)}{4 G_N} + S_\text{bulk},
\end{align}
where $\CA(\chi_A)$ is the area of the RT surface $\chi_A$ and $S_\text{bulk}$ is the bulk entanglement entropy associated to the entanglement wedge of $A$. The entanglement entropy, although not an actual observable, is therefore a quantity that depends only on the reduced density matrix of the state on region $A$, but depends on the entire entanglement wedge in the bulk. Somewhat remarkably, \eqref{eq:RT} alone is, in fact, sufficient to imply the existence of decoding channels $\CD$ that can be used for entanglement wedge reconstruction~\cite{dong2016reconstruction, cotler2017entanglement}. The key intermediate step, which was shown in~\cite{jafferis2016relative}, is that \eqref{eq:RT} implies an approximate equality between bulk and boundary relative entropies.

Unfortunately, even though it is, at this point, very well established that entanglement wedge reconstruction is possible in principle (and hence that decoding channels $\mathcal{D}$ must exist), it has proved somewhat challenging to find constructions that work for bulk operators lying outside the causal wedge (and hence for which we cannot use the HKLL prescription) and that are both explicit and practical. An explicit, if somewhat impractical, general construction was given in~\cite{almheiri2015bulk,dong2016reconstruction}. However, this construction relies on the unphysical assumption of exact quantum error correction, which does not exist at finite $N$. 

It was shown in~\cite{faulkner2017bulk} that the evolution of bulk operators in bulk modular time is related via the extrapolate dictionary to the evolution of boundary operators in boundary modular time. Since bulk modular evolution should be linear in the free field approximation $N \to \infty$, one might hope to expand a bulk operator at any point in the entanglement wedge in terms of the modular evolution of operators at the boundary of the wedge, and thus derive an explicit formula for the boundary representation of the bulk operator. However, as yet, the details of this expansion remain unknown, and it is not even clear how to show rigorously that one should exist at all.

Finally, it was demonstrated in~\cite{cotler2017entanglement}, using the tools of approximate operator algebra quantum error correction, that robust entanglement wedge reconstruction can be achieved using the so-called \emph{twirled Petz map}~\cite{junge2018universal}, even at finite $N$. The twirled Petz map is an example of a ``universal recovery channel'' -- a general purpose decoding map that lets one approximately recover from the action of a quantum channel.  That is, given a quantum channel $\mathcal{N}$ and a fixed state $\sigma$, the goal is to find a recovery channel $\mathcal{R}_{\sigma, \mathcal{N}}$ such that $\mathcal{R_{\sigma, \mathcal{N}}}\circ\mathcal{N}[\rho]\approx\rho$ for all $\rho$.  The twirled Petz map is one such recovery channel $\mathcal{R}_{\sigma, \mathcal{N}}$, defined to be
\begin{align}\label{twirledPetz}
\mathcal{R}_{\sigma, \mathcal{N}} (\rho) = \int \, dt \, \frac{\pi}{2}\,(\cosh (\pi t) + 1)^{-1} \,\,\sigma^{\frac{1 - it}{2}} \mathcal{N}^\dagger\left( [\mathcal{N}(\sigma)]^{- \frac{1 - it}{2}} \rho \,[\mathcal{N}(\sigma)]^{- \frac{1 + it}{2}}\right) \sigma^{\frac{1+ it}{2} }.
\end{align}
If we replace $\sigma$ by the maximally mixed state $\tau$, the expression simplifies significantly.  We can use the twirled Petz map for bulk reconstruction by setting the channel $\mathcal{N}$ to be $\mathcal{N} = [ J(\cdot)J^\dagger]_A$.  With the simple choice $\sigma=\tau$,  this leads to the boundary reconstruction $\Phi_A$ of a bulk operator $\phi_a$ as
\begin{align}\label{twirledOp}
\Phi_A = \mathcal{R}_{\tau, \mathcal{N}}^\dagger (\phi_a) = \frac{1}{d_{\text{code}}} \int \, dt \, \frac{\pi}{2}\,(\cosh (\pi t) + 1)^{-1} \,\, \tau_A^{- \frac{1 - it}{2}} [J \phi_a J^\dagger]_A \tau_A^{- \frac{1 + it}{2}},
\end{align}
where $\tau_A = \CN(\tau)$.  Even though choosing the reference state to be maximally mixed has simplifed the expression, it still involves a twirling or averaging over boundary modular time with the precisely chosen weighting $\pi/2\,(\cosh (\pi t) + 1)^{-1}$.

In this paper, we will show that such averaging is unnecessary for code spaces of any fixed finite dimension in the semiclassical limit $N \to \infty$ and $G_N \to 0$. Instead it is sufficient to use the much simpler \emph{Petz map}~\cite{ohya2004quantum} reconstruction
\begin{align}
\Phi_A = \frac{1}{d_{\text{code}}} \tau_A^{-1/2} [J \phi_a J^\dagger]_A \,\tau_A^{-1/2}.
\label{PetzOp}
\end{align}
We are hopeful that this more tractable recovery map should prove significantly easier to evaluate explicitly; we discuss the challenges and prospects of doing so in Section \ref{sec:discuss}. For other examples of situations where twirling is unnecessary and the Petz map itself is sufficient, see \cite{li2018squashed, alhambra2018work, alhambra2017dynamical, lemm2017information}.

Our strategy for proving the efficacy of the Petz map for entanglement wedge reconstruction builds on work by Barnum and Knill~\cite{barnum2002reversing}, who showed that, for ordinary subspace quantum error correction, the Petz map will always have an \emph{average} decoding error that is almost as small as the average error of the optimal decoding channel. Roughly speaking, the Petz map is always `pretty good'. We extend these results to subsystem and operator algebra quantum error correcting codes and then show that the average error can always be used to bound the worst-case error, so long as the dimension of the code space does not grow too quickly in the limit of large $N$. (We discuss very large code spaces such as those of black hole microstates briefly in Section \ref{sec:discuss}.)

In Section \ref{sec:entwedge}, we formalize the problem of entanglement wedge reconstruction in the language of quantum error correction and show how to adapt the results of Barnum and Knill to prove that reconstruction is possible using the Petz map.  Our main technical result is a general theorem in quantum error correction, the proof of which is given in Section \ref{fullproof}, and an application of which simplifies the problem of entanglement wedge reconstruction. Section \ref{sec:discuss} consists of a brief discussion of potential applications and extensions of our work.

\section{Entanglement Wedge Reconstruction and the Petz Map} \label{sec:entwedge}
In order to apply information-theoretic techniques to the problem of entanglement wedge reconstruction, we first need to rephrase our task in the language of quantum information. We employ the same framework used in~\cite{cotler2017entanglement} -- finite-dimensional approximate operator algebra quantum error correction.

The AdS/CFT correspondence is a duality between a boundary conformal field theory with Hilbert space $\mathcal{H}_{CFT}$, and a bulk quantum gravity theory. In principle, if AdS/CFT is supposed to be a true duality between theories, the `bulk' Hilbert space should be isomorphic to the boundary Hilbert space $\mathcal{H}_{CFT}$. However, a complete, non-perturbative, microscopic description of the entire Hilbert space from a purely bulk perspective, if one exists, remains unknown. Moreover, any such Hilbert space would be dominated by large black holes. Instead, we are normally only interested in a small `code subspace' $\mathcal{H}_{\text{code}}$ of states with a smooth semiclassical bulk geometry; for example, we might consider small bulk perturbations about the vacuum state. We therefore have an isometry $J: \mathcal{H}_{\text{code}} \to \mathcal{H}_{CFT}$. Equivalently, we can consider the quantum channel $\mathcal{J}(\cdot) = J(\cdot)J^\dagger$ which maps bulk density matrices to boundary density matrices.  As it turns out, none of our results rely on $\mathcal{J}$ being an isometry as opposed to a more general quantum channel. 

For simplicity, we assume that both $\mathcal{H}_{\text{code}}$ and $\mathcal{H}_{CFT}$ are finite-dimensional. In the case of $\mathcal{H}_{\text{code}}$, this is justified by the fact that we cannot include arbitrarily high energy excitations in the bulk without causing significant backreaction and eventually creating a black hole. In the case of $\mathcal{H}_{CFT}$, we should be able to regularize the boundary theory in the UV, while only affecting bulk physics close to the boundary. Of course, the real value of these assumptions for our purposes is that they allow us to apply known results from the large literature on finite-dimensional quantum error correction.

We denote the algebra of observables on the Hilbert space $\mathcal{H}_{\text{code}}$ by $\CB(\mathcal{H}_{\text{code}})$ and the algebra of observables on $\mathcal{H}_{CFT}$ by $\CB(\mathcal{H}_{CFT})$. The entanglement wedge $a$ has an associated von Neumann subalgebra $ \mathcal{M}_a\stackrel{i}{\hookrightarrow} \CB(\mathcal{H}_{\text{code}})$, consisting of bulk observables that act only on $a$; similarly, the boundary region $A$ is associated with a von Neumann subalgebra $\mathcal{M}_A \stackrel{i}{\hookrightarrow} \CB(\mathcal{H}_{CFT})$. We use the notation from~\cite{cotler2017entanglement}, where the space of density matrices on a von Neumann subalgebra $\mathcal{M}$ acting on a Hilbert space $\mathcal{H}$ is denoted by $S(\mathcal{M}) \cong S(\mathcal{H}) \cap \mathcal{M}$. This space is isomorphic to the space of positive normalized linear functionals on the algebra. See the appendix of~\cite{cotler2017entanglement} for more details.

The question of entanglement wedge reconstruction can then be rephrased as the question of whether the channel $\mathcal{N} = [\mathcal{J}(\cdot)]_A$ forms an approximate error-correcting code for the algebra $\mathcal{M}_a$. Here, the restriction channel $[\cdot]_A$ simply projects the density matrix onto the algebra $\mathcal{M}_A$. In other words, entanglement wedge reconstruction is possible if (and only if) there exists a decoding channel $\mathcal{D}: S(\mathcal{M}_A) \to S(\mathcal{M}_a)$ such that
\begin{align}
\mathcal{D} \circ \mathcal{N} (\rho) \approx \rho_a,
\end{align}
for all states $\rho \in S(\mathcal{H}_{\text{code}})$; the restriction $\rho_a$ is the projection of $\rho$ onto $\mathcal{M}_a$. For subsystem algebras, this corresponds to taking a partial trace over the complementary subsystem and hence agrees with the usual notion of a reduced density matrix; operator algebra quantum error correction therefore generalizes subsystem quantum error correction. 

In the Heisenberg (adjoint) picture, this condition becomes
\begin{align}
\mathcal{N}^\dagger \circ \mathcal{D}^\dagger ( \phi_a)  = \mathcal{J}^\dagger \circ \mathcal{D}^\dagger( \phi_a )\approx \phi_a.
\end{align}
Note that, since the adjoint of the restriction channel is simply the embedding of the subalgebra in the larger algebra of observables, $\mathcal{N}^\dagger(O_A) = \mathcal{J}^\dagger(O_A)$ for all operators $O_A \in \mathcal{M}_A$. In other words, $\Phi_A = \mathcal{D}^\dagger (\phi_a)$ acts in approximately the same way as $\phi_a$:
\begin{align}
\tr (\Phi_A \CJ(\rho)) \approx \tr (\phi_a \rho)
\end{align}
The complete setup, in both the Schr\"{o}dinger and Heisenberg pictures, is shown in Figure \ref{fig:squares}.

\begin{figure}
\centering
	\begin{tikzcd}[column sep=5em, row sep=3em]
	  \mathcal M_a \arrow[r, yshift=0.7ex, hookrightarrow, "\text{$i_a$}"]
	  \arrow[d, dashed, "\mathcal D^\dagger", swap] &
	  \mathcal B(\CH_\text{code}) 
	  \\
	  \mathcal M_A \arrow[r, swap, hookrightarrow, "\text{$i_A$}"] 
	  \arrow[ru, "\mathcal N^\dagger"]
	  &
	  \mathcal B(\CH_\text{CFT}) \arrow[u, swap, "\mathcal J^\dagger"]
	\end{tikzcd}
	\qquad\qquad
	\begin{tikzcd}[column sep=3em, row sep=3em]
	  S(\mathcal M_a) 
	  &
	  \arrow[l, yshift=0.7ex, "\text{$res_a$}", swap] S(\mathcal H_\text{code}) \arrow[d, "\mathcal J"] \arrow[dl, "\mathcal N"] 
	  \\
	  S(\mathcal M_A) \arrow[u, dashed, "\mathcal D"] &
	  \arrow[l, "\text{$res_A$}"] S(\mathcal H_\text{CFT})
	\end{tikzcd}
	\caption{In the Heisenberg picture, $\CM_a\stackrel{i}{\hookrightarrow}\CB(\CH_\text{code}) $ and $\CM_A\stackrel{i}{\hookrightarrow}\CB(\CH_{CFT})$ are von Neumann subalgebras acting on the code space $\CH_\text{code}$ and CFT Hilbert space $\CH_{CFT}$ respectively. The Heisenberg channel $\CJ^\dagger = J^\dagger(\cdot)J$ maps boundary observables to their projection in the code space. The task of entanglement wedge reconstruction is to find a Heisenberg decoding channel $\CD^\dagger : \CM_a \to \CM_A$ that maps bulk observables $\phi_a$ in $\CM_a$ to boundary observables $\Phi_A$ in $\CM_A$. When projected into the code space using $\CJ^\dagger$, the boundary observable $\Phi_A$ should reproduce the original bulk observable $\phi_a$. In the Schr\"{o}dinger picture, the directions of all channels are reversed. The channel $\CJ$ now maps bulk states to the corresponding boundary states. The Heisenberg channels $i_a$ and $i_A$, which embed the von Neumann subalgebras $\CM_a$ and $\CM_A$ into the larger algebras of observables $\CB(\CH_\text{code})$ and $\CB(\CH_{CFT} )$, are the adjoints of the restriction maps onto $S(\CM_a)$ and $S(\CM_A)$ respectively. Finally, the decoding channel $\CD: S(\CM_A) \to S(\CM_a)$ satisfies $\CD[\CJ(\cdot)_A] \approx (\cdot)_a$. }
	\label{fig:squares}
\end{figure}

It was argued in~\cite{cotler2017entanglement} that the twirled Petz map provides an example of a decoding map with an error that is perturbatively suppressed in $1/N$. It was then shown in~\cite{hayden2018learning} that there must exist some decoding channel $\mathcal{D}$ with a non-perturbatively small error; however, this argument was non-constructive. Both results relied heavily on the approximate equality between bulk and boundary relative entropies found in~\cite{jafferis2016relative}. A refined statement of this approximate relative entropy equality was derived in~\cite{dong2018entropy}, which allows one to show the existence of a decoding channel that is accurate to all orders in perturbation theory. Here, we shall simply take as our starting assumption the existence of \emph{some} good decoding channel $\CD'$; we will not need to know any details of this channel. We can therefore use the result of~\cite{hayden2018learning} to assume that the decoding error when using this channel is non-perturbatively small. The following theorem, which we prove in Section~\ref{fullproof}, then implies that the Petz map is also a good decoding channel:
\begin{thm}
\label{mainthm}
Let $\CM_a \stackrel{i}{\hookrightarrow}  \CB(\CH_\text{code})$ be a von Neumann subalgebra acting on the code space $\CH_\text{code}$ with dimension $d_{\text{code}}$, let $\CN$ be a quantum channel, and suppose that there exists a channel $\CD'$ such that $$\lVert\CD'\circ\CN(\rho) - \rho_{a}\rVert_1< \delta$$. Let 
$$\CP_{\tau, \CN} := \frac{1}{d_{\text{code}}} \CN^\dagger\left[\CN(\tau)^{-1/2}(\cdot)\CN(\tau)^{-1/2}\right]$$ 
be the Petz map with maximally mixed reference state $\tau$. Then 
\begin{align}
\left\lVert\CP_{\tau, \CN} \circ \CN (\rho)|_a-\rho_a\right\rVert_1 \le d_{code}\sqrt{8\delta}.
\end{align}
\end{thm}
Note that our bound on the error when reconstructing the reduced state using the Petz map $\CP_{\tau, \CN}$ is significantly higher than the original error $\delta$. Not only is the error proportional to $\sqrt{\delta}$, but it is also proportional to the dimension of the code space $d_{\text{code}}$. As we shall see in Section~\ref{fullproof}, the square root appears because of inefficiencies in converting between trace distances and fidelities using the Fuchs-van de Graaf inequalities~\cite{fuchs1999cryptographic}, while the factor of $d_\text{code}$ appears in order to convert a bound on the average-case error into a bound on the worst-case error.  Nevertheless, so long as the error using the original decoding channel $\mathcal{D}'$ is non-perturbatively small, the Petz map error will also be non-perturbatively small, provided the dimension of the code space does not grow superpolynomially in $N$. For most code spaces of interest, such as perturbations about the vacuum, the code space dimension will be $O(1)$, and so this factor of code space dimension is not problematic. We discuss very large code spaces, such as those containing large numbers of black hole microstates, briefly in Section \ref{sec:discuss}. However, so long as we confine ourselves to perturbative excitations of quantum fields in a fixed gravitational background, the Petz map can always be trusted -- no twirling is required.

\section{Proof of Theorem \ref{mainthm}}
\label{fullproof}
The spirit of Theorem \ref{mainthm} follows that of \citet{barnum2002reversing}, who proved the following theorem for ordinary subspace quantum error correction:
\begin{thm}[Barnum-Knill~\cite{barnum2002reversing}]
\label{originalBK}
Given any pair of quantum channels $\CD'$, $\CN$, and ensemble of commuting density matrices $(p_k,\rho_k)$ whose sum $\sum_k p_k\rho_k = \rho$, the Petz map 
$$
\CP_{\rho,\CN}[\cdot] :=\rho^{1/2}\CN^\dagger\left[\CN(\rho)^{-1/2}(\cdot)\CN(\rho)^{-1/2}\right]\rho^{1/2}
$$
with reference state $\rho$, satisfies
\begin{align} \label{eq:BK}
\sum_{k} p_k F(\rho_k, \CP_{\CN,\rho}\circ\CN)\ge (\sum_{k} p_kF(\rho_k, \CD'\circ\CN))^2.
\end{align}
Here, the entanglement fidelity $F(\sigma,\CZ)$ is defined by
$$F(\sigma,\CZ) :=  \bra{\sigma}V_{\CZ}^\dagger \left(\ket{\sigma}\bra{\sigma} \otimes \mathds{1}_E \right) V_{\CZ}\ket{\sigma},$$
where $\ket{\sigma} \in \mathcal{H}_\text{code} \otimes \mathcal{H}_R$ is a purification of $\sigma \in S(\mathcal{H}_\text{code})$ and  $V_{\CZ}: \mathcal{H}_\text{code} \to \mathcal{H}_\text{code} \otimes \mathcal{H}_E$ is a Stinespring extension of the channel $\CZ: S(\mathcal{H}_\text{code}) \to S(\CH_\text{code})$.  

\end{thm}

If we now assume that $\CD'$ is a recovery channel for $\CN$ that works with high fidelity, then Theorem~\ref{originalBK} states that $\CP_{\CN,\rho} \circ \CN$ is close to the identity when measured using the average entanglement fidelity of an ensemble $\{\rho_k\}$ with average state $\rho$. Note that, unlike our Theorem \ref{mainthm}, there is no factor of $d_{\text{code}}$ in the size of the error for the Petz map $\CP_{\CN,\rho}$ as compared to the original decoding channel $\CD'$. Instead, \eqref{eq:BK} implies that the error, measured using the \emph{average entanglement fidelity}, increases by at most a factor of two.\footnote{An entanglement fidelity $F(\rho, \CD \circ \CN) = 1$ implies perfect recovery of a purification of $\rho$. Hence, we can naturally quantify the recovery error when decoding using the channel $\CD'$ by $$\delta  = 1 - \sum_{k} p_kF(\rho_k, \CD'\circ\CN).$$ The equivalent error measure, when decoding using the Petz map $\CP_{\rho,\CN}$, will then be bounded by $2\,\delta$.} The factor of $d_{\text{code}}$ will appear when we convert this average error into a worst-case error.

For concreteness, let us write down an explicit basis for the von Neumann subalgebra $\CM_a$. The exact description of $\CJ$ and $\CM_A$ (and hence $\CN$) are unimportant for our purposes. It is a fact about finite-dimensional von Neumann algebras (see, for example,~\cite{harlow2017ryu}) that we can always find a set of Hilbert spaces $\mathcal{H}_{\alpha}$ and $\mathcal{H}_{\bar \alpha}$, parameterized by $\alpha$, such that
\begin{align}
\begin{split}
\CM_{a}&= \bigoplus_\alpha \CB(\CH_{\alpha}) \otimes \mathds{1}_{\bar{\alpha}},\\
\CH_{code} &= \bigoplus_{\alpha} \CH_\alpha\otimes\CH_{\bar{\alpha}}.
 \end{split}
\end{align}
Note that
\begin{equation} \label{eq:dims}
\sum_\alpha d_\alpha d_{\bar{\alpha}} = d_\text{code},
\end{equation}
where $d_\alpha$, $d_{\bar{\alpha}}$ and $d_\text{code}$ are the dimensions of $\CH_\alpha$, $\CH_{\bar{\alpha}}$ and $\CH_\text{code}$ respectively. In this basis, any state $\rho_a \in S(\CM_a)$ can be parameterized as
\begin{equation} \label{eq:rho_a}
\rho_a = \sum_\alpha p_\alpha \rho_\alpha \otimes \tau_{\bar{\alpha}} = \sum_{\alpha,i_\alpha} p_\alpha p_{i_\alpha}^{(\alpha)} \ket{i_\alpha}\bra{i_\alpha} \otimes \tau_{\bar{\alpha}},
\end{equation}
where the states $\tau_{\bar{\alpha}} \in S(\CH_{\bar{\alpha}})$ are maximally mixed, $\rho_\alpha \in S(\CH_\alpha)$ are normalized density matrices, $p_\alpha$ and $p_i^{(\alpha)}$ are normalized probability distributions, and $\ket{i_\alpha}$ forms an orthonormal basis for $\mathcal{H}_\alpha$. 

We now have all the ingredients we need to begin a proof of Theorem \ref{mainthm}. Let $\CZ = \CP_{\tau,\CN} \circ \CN$. We first note that $\CZ$ is a self-adjoint superoperator. For any operator $\phi$,
\begin{align}
\tr[\phi \,\CZ(\rho)] &=  \frac{1}{d_\text{code}} \tr\left[\phi\,\CN^\dagger\left(\CN(\tau)^{-1/2}\,\CN(\rho)\,\CN(\tau)^{-1/2}\right)\right]\\
&= \tr[\CZ(\phi) \rho] = \tr[\phi \,\CZ^\dagger(\rho)].
\end{align} 
Hence we have that $\CZ = \CZ^\dagger$. Note that this argument relied crucially on our choice for the reference state in the Petz map $\CP_{\tau,\CN}$ as the maximally mixed state.

Now, let $\phi_a \in \CM_a$ be a Hermitian operator, which we can assume to have eigendecomposition 
\begin{align}
\phi_a= \lambda_{i_\alpha}\ket{i_\alpha}\bra{i_\alpha}.
\end{align}
We can bound the operator norm
\begin{align} \label{eq:step1}
 \lVert\CZ^\dagger(\phi_a) - \phi_a \rVert_\infty &\le \lVert\CZ^\dagger(\phi_a) - \phi_a \rVert_1 \\
&\le \sum_{\alpha, i_\alpha} | \lambda_{i_\alpha}|\  \lVert(\CZ^\dagger- \mathds{1}) [\ket{i_\alpha}\bra{i_\alpha} \otimes \mathds{1}_{\bar{\alpha}}]\rVert_1\\
&=  \sum_{\alpha, i_\alpha} | \lambda_{i_\alpha}|\ d_{\bar{\alpha}} \lVert\CZ[\rho^{i_\alpha}_{\alpha}]-\rho^{i_\alpha}_{\alpha}\rVert_1,
\end{align}
where the first inequality uses the monotonicity of the Schatten p-norms, the second inequality used the triangle inequality, and in the final equality we factored out $d_{\bar{\alpha}}$ so that $\rho^{i_\alpha}_{\alpha} = \ket{i_{\alpha}}\bra{i_\alpha}\otimes \tau_{\bar{\alpha}}$ are normalized states, and more importantly we used the fact that the channel $\CZ$ is self-adjoint. We now simply need to bound the average trace norm error of the channel $\CZ$ on states $\rho_a \in S(\CM_a)$. This is quadratically controlled by Theorem~\ref{originalBK}:

\begin{prop}\label{average1}
\begin{align}
\sum_{i_\alpha,\alpha} \frac{d_{\bar{\alpha}}}{d_{code}}  \lVert\CZ[\rho^{i_\alpha}_{\alpha}]-\rho^{i_\alpha}_{\alpha}\rVert_1^2 \le 8\delta
\end{align}
\end{prop}
\begin{proof}
We first note that
\begin{equation} \label{eq:mix}
\sum_{i_\alpha,\alpha} \frac{d_{\bar{\alpha}}}{d_\text{code}} \rho^{i_\alpha}_\alpha = \tau.
\end{equation}
Hence
\begin{align}
\sum_{i_\alpha,\alpha} \frac{d_{\bar{\alpha}}}{d_{code}}  \lVert\CZ[\rho^{i_\alpha}_{\alpha}]-\rho^{i_\alpha}_{\alpha}\rVert_1^2 & \le 4\sum_{ i_\alpha,\alpha} \frac{d_{\bar{\alpha}}}{d_{\text{code}}}(1-F(\rho^{i_\alpha}_{\alpha} ,\CZ[\rho^{i_\alpha}_{\alpha} ]))\\
&\le 4 -  4(\sum_{\alpha} \frac{d_{\bar{\alpha}}}{d_{\text{code}}}F(\rho^{i_\alpha}_{\alpha}, \CD'\circ\CN[\rho^{i_\alpha}_{\alpha}]))^2\\
&\le 4 -  4\left(\sum_{\alpha} \frac{d_{\bar{\alpha}}}{d_{\text{code}}}\left(1-\frac{1}{2}\lVert\CD'\circ\CN[\rho^{i_\alpha}_\alpha]-\rho^{i_\alpha}_\alpha\rVert_1\right) \right)^4\\
&\le 8\delta,
\end{align}
where the first inequality uses one of the Fuchs-van de Graaf inequalities~\cite{fuchs1999cryptographic}, the second uses \eqref{eq:mix} and Theorem~\ref{originalBK}, the fourth again uses the Fuchs-van de Graaf inequalities, and the fifth uses our assumption $\lVert\CD'\CN(\rho) - \rho_{a}\rVert_1< \delta$ and \eqref{eq:dims}.
\end{proof}
Applying Proposition \ref{average1} to \eqref{eq:step1}, we find 
\begin{align}
 \sum_{\alpha, i_\alpha} | \lambda_{i_\alpha}| \ d_{\bar{\alpha}} \lVert\CZ[\rho^{i_\alpha}_{\alpha}]-\rho^{i_\alpha}_{\alpha}\rVert_1&\leq \lVert \phi_a \rVert_\infty \sum_{\alpha, i_\alpha} \sqrt{d_{\bar{\alpha}} d_\text{code}} \cdot \sqrt{\frac{d_{\bar{\alpha}}}{d_\text{code}}}   \lVert\CZ[\rho^{i_\alpha}_{\alpha}]-\rho^{i_\alpha}_{\alpha}\rVert_1\\
&\le \lVert \phi_a \rVert_\infty\sqrt{\sum_{\alpha, i_\alpha} d_{\bar{\alpha}} d_\text{code}}  \cdot \sqrt{8\delta}\\
& = \lVert\phi_a\rVert_\infty \,d_\text{code}  \cdot \sqrt{8\delta},
\end{align}
where, in the first inequality, we used the fact that $\lVert \phi_a \rVert_\infty \geq |\lambda_{i_\alpha}|$ for all $\lambda_{i_\alpha}$ and, in the second inequality, we used the Cauchy-Schwarz inequality. We therefore find that
$$\lVert\CZ^\dagger(\phi_a) - \phi_a\rVert_\infty \le \lVert\phi_a\rVert_\infty d_{code} \sqrt{8\delta}.$$
Now, since
\begin{align}
\lVert \CZ(\rho)_a - \rho_a \rVert_1 &=
\sup_{\phi_a \in \CB(\CH_a)} \frac{1}{\lVert \phi_a \rVert_\infty} \tr([\CZ^\dagger(\phi_a) - \phi_a]\rho) \leq  \sup_{\phi_a \in \CB(\CH_a)} \frac{1}{\lVert \phi_a \rVert_\infty} \lVert\CZ^\dagger(\phi_a) - \phi_a\rVert_\infty,
\end{align}
we immediately arrive at our desired result
\begin{align} \label{eq:final}
\lVert(\CP_{\tau,\CN}\circ \CN[\rho])_a-\rho_a\rVert_1=\lVert(\CZ[\rho])_a-\rho_a\rVert_1 \le d_{code} \sqrt{8\delta},
\end{align}
for any state $\rho \in S(\CH_\text{code})$.

Note that we could have directly seen from Proposition \ref{average1} (using the triangle inequality) that for any state $\rho_a \in S(\mathcal{M}_a)$ we have
\begin{align} \label{eq:rhoa_only}
\lVert \CZ(\rho_a)_a - \rho_a \rVert_1 \leq \sqrt{8\delta d_\text{code}}.
\end{align}
However, although this is a tighter bound than \eqref{eq:final}, it only applies to states in the code space that are of the form given in \eqref{eq:rho_a}. In the Heisenberg picture, we want our reconstructed operator to work for all states in the code space -- not just states in $S(\CM_a)$.

The same problem of extending reconstruction from states $\rho_a \in S(\CM_a)$ to all states $\rho\in S(\CH_\text{code})$ was previously encountered for the twirled Petz map in~\cite{cotler2017entanglement}. It was shown that the approximate equality between bulk and boundary relative entropies~\cite{jafferis2016relative} implies that any state $\rho \in S(\CH_\text{code})$ satisfies
\begin{align} \label{eq:comp}
\CN(\rho) \approx \CN(\rho_a).
\end{align}
Hence \eqref{eq:rhoa_only} implies that, for all states $\rho \in S(\CH_\text{code})$, we have
\begin{align} \label{eq:holo}
\lVert \CZ(\rho)_a - \rho \rVert_1 \leq \sqrt{8\delta d_\text{code}} + \varepsilon,
\end{align}
where $\varepsilon$ is independent of $d_\text{code}$ and $\varepsilon \to 0$ in the limit $N\to\infty$. 
However, \eqref{eq:comp} really only holds because of the complementary recovery property of AdS/CFT. Not only does region $A$ learn everything about the entanglement wedge $a$, it also learns nothing about the complementary bulk region $\bar a$, which is the entanglement wedge of region $\bar A$. In general, operator algebra quantum error correcting codes will not even approximately satisfy \eqref{eq:comp} -- consider, for example, the case where $\CN$ is the identity channel and $\CM_a$ is any proper subalgebra of the algebra of observables $\CB(\CH_\text{code})$. It follows that \eqref{eq:holo} is specific to holographic codes. In contrast, Theorem \ref{mainthm} is a completely general fact about operator algebra quantum error correction. \textbf{Theorem \ref{mainthm} is therefore a true extension of the range of validity of the Petz map as a general-purpose, approximate recovery channel to operator algebra and subsystem codes.}

\section{Discussion} \label{sec:discuss}
In this work, we proved a theorem in quantum error correction about the quality of decoding using the Petz map as a general recovery channel.  Our theorem generalizes the work of~\citet{barnum2002reversing} to the case of operator algebra quantum error correction, and subsystem quantum error correction.

By applying our theorem to AdS/CFT, we showed that entanglement wedge reconstruction can be achieved using the Petz map, so long as the dimension of the code space we expect to be able to reconstruct is not too large. In particular, the Petz map constitutes a good recovery map provided the code space dimension does not grow superpolynomially in the limit of large $N$. In practice, this is almost always the case for code spaces of interest. 

It is worth commenting briefly on the major exception to this rule: code spaces containing large numbers of black hole microstates\footnote{For a detailed discussion of this topic see, for example,~\cite{hayden2018learning}.}. The entropy of such code spaces may, in general, be $O(1/G_N)$. Hence the dimension of the code space may be exponential in $N$. However, as yet, the only black hole microstates that we understand are generic, equilibrium microstates. For code spaces made out of such microstates, we would expect the worst-case and average reconstruction errors to approximately agree, even though the in-principle large code space dimension means that very large differences between these two fidelities are possible. It is therefore reasonable to hope that the Petz map will still be valid for entanglement wedge reconstruction. On the other hand we should not trust any semi-classical description of non-generic, finely-tuned black hole microstates, and thus entanglement wedge reconstruction might not be possible for such states. As such, there are no known situations in which entanglement wedge reconstruction is possible, yet we cannot use the Petz map to achieve it.

While we emphasized the utility of the Petz map over other reconstruction techniques, we have not made any serious attempt to actually evaluate the Petz map in particular cases. Though the Petz map is much simpler to write down and, in principle, evaluate than the twirled Petz map, there still remain significant obstacles to doing so. Let us briefly discuss the challenges involved. We wish to explicitly evaluate
\begin{align}
\Phi_A = \frac{1}{d_{code}} \tau^{-1/2}_A [J \phi_a J^\dagger]_A \,\tau_A^{-1/2}.
\end{align}
The operator $ J \phi_a J^\dagger$ can be found by taking the global HKLL boundary reconstruction $\Phi^\text{HKLL}$ and projecting it into the code space~\cite{cotler2017entanglement}
\begin{align}
 J \phi_a J^\dagger = P_\text{code} \Phi^\text{HKLL} P_\text{code}.
\end{align}
Therefore, the main challenge lies in finding the restriction of this operator to region $A$. For simplicity, we assume, in accordance with common practice (though not with reality) that the CFT Hilbert space factorizes as $\CH_{CFT} \cong \CH_A \otimes \CH_{\bar{A}}$ with $\CM_A \cong \CB(\CH_A)$; the restriction map is then simply a partial trace over $\CH_{\bar{A}}$. One difficulty arises because the HKLL procedure gives an operator $\Phi$ that is not localized in time. To take the partial trace over region $\bar A$, we need to use the Heisenberg equations of motion to rewrite $\Phi_A$ in terms of operators at time zero\footnote{For a more detailed discussion of similar issues, see~\cite{cotler2017entanglement}.}. Such operators will in general be very complicated and hard to evaluate. Essentially, the obstruction is simply the usual obstruction to evaluating quantities that are not protected by symmetry on the boundary side of AdS/CFT. Strongly coupled quantum field theories are simply hard to deal with; thankfully, there also exists a weakly coupled bulk. 
\section{Acknowledgements}
We would like to thank Patrick Hayden, Richard Nally, and Michael Walter for valuable discussions. We would also like to thank Howard Barnum for insightful and stimulating conversation in the early stages of this project. CFC was supported by the Physics/Applied Physics/SLAC Summer Research Program for undergraduates at Stanford University. GP is supported by the Simons Foundation ``It from Qubit'' collaboration, AFOSR grant number FA9550-16-1-0082 and {DOE award DE-SC0019}. GS was supported by an IQIM postdoctoral fellowship at Caltech, and by the Stanford Institute for Theoretical Physics.

\bibliographystyle{unsrtnat}
\bibliography{biblio}
\end{document}